\documentclass[10pt,reqno]{amsart}
\usepackage{latexsym}
\usepackage{amsmath}
\usepackage{cases}
\usepackage{amssymb, amscd, amsthm, multirow}
\usepackage[all]{xy}
\usepackage{verbatim}
\usepackage{bbm}
\usepackage{mathrsfs}
\usepackage{booktabs}
\usepackage{threeparttable}
\usepackage{color}

\newtheorem{corollary}{{Corollary}}
\newtheorem{theorem}{{Theorem}}

\newtheorem{definition}{{Definition}}

\newtheorem{lemma}{{Lemma}}
\newtheorem{remark}{{Remark}}

\newtheorem{proposition}{{Proposition}}

\newtheorem*{conjecture}{{Conjecture}}

\newtheorem{example}{{Example}}

\newcommand{\F}{\mathbb{F}}

\newcommand{\Tr}{\mathrm{Tr}}
\newcommand{\Z}{\mathbb{Z}}

\renewcommand{\L}{\mathcal{L}}

\newcommand{\bU}{\mathbf{U}}
\newcommand{\bA}{\mathbf{A}}
\newcommand{\N}{\mathcal{N}}
\newcommand{\C}{\mathcal{C}}

\newcommand{\rank}{\mathrm{rank}}

\renewcommand{\Im}{\mathrm{Im}}

\usepackage[shortlabels, inline]{enumitem}
\SetEnumerateShortLabel{a}{\textup{(\alph*)}}
\SetEnumerateShortLabel{A}{\textup{(\Alph*)}}
\SetEnumerateShortLabel{1}{\textup{(\arabic*)}}
\SetEnumerateShortLabel{i}{\textup{(\roman*)}}
\SetEnumerateShortLabel{I}{\textup{(\Roman*)}}

\begin{document}
\title[Vectorial functions with maximal number of bent components]{On vectorial functions with  maximal number of bent components}

\author{Xianhong Xie$^{1}$}
\address{$^1$School of Information and Computer, Anhui Agricultural University,
 Hefei 230036, China}
\email{xianhxie@ahau.edu.cn}

\author{Yi Ouyang$^{2,3}$}
\address{$^2$School of Mathematical Sciences, CAS Wu Wen-Tsun Key Laboratory of Mathematics,
University of Science and Technology of China, Hefei 230026, China}
\address{$^3$Hefei National Laboratory, Hefei 230088, China}
\email{yiouyang@ustc.edu.cn}
\thanks{Partially supported by Innovation Program for Quantum Science and Technology (Grant No. 2021ZD0302902) and Anhui Initiative
in Quantum Information Technologies (Grant No. AHY150200)}

\begin{abstract} We study  vectorial functions with maximal number of bent components in this paper. We first study the Walsh transform and nonlinearity of $F(x)=x^{2^e}h(\Tr_{2^{2m}/2^m}(x))$, where $e\geq0$ and $h(x)$ is a permutation over $\F_{2^m}$. If $h(x)$ is monomial, the nonlinearity of $F(x)$ is shown to be at most $ 2^{2m-1}-2^{\lfloor\frac{3m}{2}\rfloor}$ and some non-plateaued  and plateaued functions attaining the upper bound are found. This gives a partial answer to the open problems proposed by Pott et al. and Anbar et al. If $h(x)$ is linear, the exact  nonlinearity of $F(x)$ is determined.
Secondly, we give a  construction of vectorial functions with maximal number of bent components from known ones, thus obtain two new classes from the Niho class and the Maiorana-McFarland class. Our construction gives a partial answer to an open problem proposed by Pott et al., and also contains vectorial functions outside the complete Maiorana-McFarland class. Finally, we show that the vectorial function $F: \F_{2^{2m}}\rightarrow \F_{2^{2m}}$, $x\mapsto x^{2^m+1}+x^{2^i+1}$ has maximal number of bent components if and only if $i=0$.
\smallskip

\end{abstract}
\maketitle

\section{Introduction}%

Bent functions, as a special class of Boolean functions, were introduced by Rothaus \cite{rt} and have been extensively studied (see \cite{ms}-\cite{boo},\cite{12}) due to their important applications in cryptography, coding theory and combinatorics.

Let $F: V\rightarrow W$ be a vectorial function, where $V$ and $W$ are $\F_2$-vector spaces of dimension $n$ and $k$ respectively.  Fix a non-degenerate inner product $\langle\  ,\  \rangle_W $ on $W$. Then for any $0\neq w\in W$, the component function $F_w$ is the Boolean function $ V\rightarrow \F_2$, $v\mapsto \langle w, F(v)\rangle_W$.  In the study of vectorial functions,  the following questions arise naturally: given $n$ and $k$, what is the maximal number of bent components of vectorial functions  and how can one construct  functions attaining this bound?

Certainly this maximal number $\leq 2^k-1$. A vectorial function with $2^k-1$ bent components is called vectorial bent.
Nyberg \cite{11} showed that  vectorial bent functions can only exist if $n$ is even and $n\geq 2k $, and presented two different constructions of such functions from known classes of bent functions.
For $k=n $,  Pott et al. \cite{5} proved that the number of bent components is at most $2^n-2^{\frac{n}{2}}$ and presented a class of binomial functions $x^{2^e}(x+x^{2^{\frac{n}{2}}})$, $0\leq e\leq m-1$ attaining the upper bound. For $k\geq \frac{n}{2} $,  Zheng et al. \cite{zheng} generalized the bound in \cite{5} and showed that this number is at most $2^k-2^{k-\frac{n}{2}}$.

Now assume $V=W$ and $n=2m$. Suppose $F:V\rightarrow V$ is a vectorial function with maximal number (i.e. $2^n-2^m$) of  bent components.

Pott et al. \cite{5} found two classes of vectorial functions of this kind, namely $x^{2^m+1}$ and $x^{2^e}(x+ x^{2^m})$. Three  new classes of the form $x\ell(x)$ where $\ell$ is a linear mapping over $\F_{2^n}$ have been found since then (see \cite{mzhang, zheng, anb, anb21}). It is an important subject to construct new vectorial functions of this kind,  in particular of the form $x\ell(x)$.

Another subject to study is the nonlinearity $\N_F$ of  these $F$.
Pott et al. \cite{5}  proved that   $N_{x^{2^e}(x+x^{2^m})} = 2^{n-1}-2^{\lfloor\frac{3n}{4}\rfloor}$ if $\gcd(e,m)=1$. If $F$ is moreover a plateaued function over $\F_{2^n}$, Anbar et al.  \cite{anb} proved that  $N_F\leq 2^{n-1}-2^{\lfloor\frac{3n}{4}\rfloor}$. Is there other $F$ attaining this bound? Does this bound hold for other $F$ with maximal number of bent components? These are also interesting problems.

We shall work on vectorial functions $F: \F_{2^n}\rightarrow \F_{2^n}$ in this paper.  Our main contributions  are the followings.
\begin{enumerate}[A]
	\item  We study the  nonlinearity of $F(x)=x^{2^e}h(\Tr_{2^n/2^m}(x))$ with $h(x)$ a permutation over $\F_{2^m}$, which were shown by Zheng et al.  \cite{zheng} to have $2^n-2^m$ bent components.
\begin{itemize}
	\item [(A.1)]If $h(x)$ is a monomial	in the Niho exponents case or the three-valued Walsh transform case, we show that
	 \[ \N_F\leq 2^{n-1}-2^{\lfloor\frac{3n}{4}\rfloor} \]
and give examples of  plateaued  and non-plateaued functions  attaining the upper bound. Therefore,  we obtain a partial answer to problems in \cite{5,anb}.
  \item [(A.2)] If $h(x)$ is a linear permutation, we show that
   \[\N_F= 2^{n-1}-2^{n-\frac{r}{2} -1},\]
   where  $ r=\min \{\rank(\Tr_{2^m/2}(ax^{2^e}h(x)))\mid a\in\F_{2^m}^*\}$.
\end{itemize}

	\item   We  construct two new classes of vectorial functions  from $\F_{2^n}$ to itself with maximal number of bent components via the Niho quadratic function and the Maiorana–McFarland class.

\item  We prove that the binomial vectorial function $F(x)=x^{2^m+1}+x^{2^i+1}$ ($0\leq i\leq m-1$) has $2^n-2^m$ bent components if and only if $i=0$.
\end{enumerate}

\section{Preliminaries}
\subsection{Basic Notations}
For $i\in \Z$, let $v_2(i)$ be the $2$-adic valuation of $i$.

For positive integers $k\mid n$,  the trace function from $\F_{2^n}$ to
its subfield $\F_{2^k}$ is the map $\Tr_{2^n/2^k}(x)=\sum\limits_{i=0}^{\frac{n}{k}-1}x^{2^{ki}}$.

For a finite dimensional $\F_2$-vector space $V$,  we always fix a non-degenerate inner product  $\langle\ ,\ \rangle=\langle\ ,\ \rangle_V$ on $V$. In particular, if $V=\F_2^n$, we let
\[\langle(v_i), (w_i) \rangle=\sum\limits_{i=1}^n v_i w_i. \]
If $V=\F_{2^n}$,  let
\[ \langle \omega, x \rangle=\Tr_{2^n/2}(\omega x). \]

For  a subspace $U$ of $V$, let $U^\perp =\{v\in V:\ \langle v,u\rangle=0\ \text{for all } u\in U\}$ be the orthogonal complementary of $U$, then  $\dim U^\perp=\dim V-\dim U$.

\subsection{Bent and plateaued functions} In this section assume $V$ is an $n$-dimensional $\F_2$-vector space, equipped with a non-degenerate inner product  $\langle\ ,\ \rangle$.

\begin{definition}\label{definition:1} Suppose  $f: V\rightarrow\F_2$ is a Boolean function.
\begin{enumerate}[i]
	\item  The Walsh transform of $f$ is
	\begin{equation*}
		W_f(w):=\sum_{v\in V}(-1)^{f(v)+\langle w, v\rangle},\ w\in V.
	\end{equation*}
	
	\item If $W_f(w)=\pm2^{\frac{n}{2}}$ for all $w\in V$, then $f(x)$ is called a bent function. In this case, its dual  function  $f^*$ is defined via the equality
	\[W_f(w)=2^{\frac{n}{2}}(-1)^{f^*(w)}.\]
	
	\item If $W_f(w)\in\{0,\pm2^{\frac{n+k}{2}}\}$, where $k\in\Z$ and $k\equiv n\pmod2$, then $f(x)$ is called a $k$-plateaued function.
	
	\item The nonlinearity of $f(x)$ is
	\[\N_f:=2^{n-1}-\frac{1}{2}\max_{\omega\in V}|W_f(\omega)|.\]
\end{enumerate}
\end{definition}

By definition, $\N_f\leq 2^{n-1}-2^{\frac{n}{2}-1}$ with equality if and only if $f$ is bent. In addition, the following result about bentness is well-known:

\begin{lemma}\label{lemma:diff} A Boolean function $f: V\rightarrow\F_2$ is bent if and only if its first derivative
	\[ D_a f(v):=f(v+a)+ f(v) \]
in the direction of $a$  is balanced for all $0\neq a\in V$.	
\end{lemma}
\subsection{Vectorial functions and bent components} Assume $W$ is a finite dimensional $\F_2$-vector space  equipped with a non-degenerate inner product  $\langle\ ,\ \rangle_W$.
\begin{definition}\label{definition:2} Suppose $F: V\rightarrow W$ is a vectorial function.
	\begin{enumerate}[i]
		\item The component function $F_w$ of $F$ at $w\in W$ is the Boolean function
		\[ F_w: V\rightarrow \F_{2},\ v\mapsto \langle w, F(v) \rangle_W. \]
		\item The Walsh transform of $F$  is
		\begin{equation*}
			W_F(a, \omega):=W_{F_a}(\omega)=\sum_{v\in V}(-1)^{F_a(v)+\langle \omega, v\rangle}, \quad a\in W-\{0\},\ \omega \in V.
		\end{equation*}
	\item $F$ is called bent if  its component functions $F_a$ for all $a\in W-\{0\}$ are   bent.
	\item $F$ is called plateaued if its component functions $F_a$ for all $a\in W-\{0\}$ are plateaued (not necessary with the same integer $k$).
	\item The nonlinearity of $F$ is the minimal nonlinearity among its component functions, i.e.,
	\[\N_F:=2^{n-1}-\frac{1}{2}\max_{a\in W-\{0\}}\max_{\omega\in V}|W_{F_a}(\omega)|.\]
	\end{enumerate}
\end{definition}
By Definition~\ref{definition:2}, a vectorial function $F:V\rightarrow W$ has at most $2^{\dim W}-1$ bent components. By the following proposition,  a vectorial function $F: V\rightarrow V$ has at most $2^n-2^{\frac{n}{2}}$ bent components, with equality only if $m:=\frac{n}{2}\in \Z$.
\begin{proposition}\label{proposition:p1} 	For $F: V\rightarrow V$, set
	\begin{equation*}S_F:=\{v\in V: \ F_v \ \text{is not\ bent}\}.\end{equation*}
	Then
	\begin{enumerate}[1]
		\item $($Pott et al. \cite{5}$)$  $|S_F|\geq 2^m$, and  $|S_F|=2^m\ ($hence $n$ is even$)$ if and only if $S_F$ is an $m$-dimensional $\F_2$-subspace of $V$.
		\item $($Hu et al. \cite{hu}$)$  If moreover  $V=\F_{2^n}$ and $|S_F|=2^m$, then $S_F=\F_{2^m}$.
	\end{enumerate}
\end{proposition}

For the nonlinearity of $F$ with maximal number of bent components, we have
\begin{lemma}[\cite{anb}] \label{lemma:anb} Suppose $n=2m$. If  $F:V\rightarrow V$ is a plateaued function with maximal number of bent components, then
	\begin{equation} \label{eq:nfbound} \N_F\leq 2^{n-1}-2^{\lfloor\frac{3n}{4}\rfloor}.\end{equation}
\end{lemma}

\section{Nonlinearity of Functions of the form $x^{2^e}h(\Tr_{2^n/2^m}(x))$}
Assume $n=2m$, $h(x)$ is a permutation over $\F_{2^m}$ and $e\geq 0$. Let
\[ \begin{split}
	F(x)=&x^{2^e}h(\Tr_{2^n/2^m}(x)): \F_{2^n}\rightarrow \F_{2^n};\\
	 H(x)=&x^{2^e}h(x):\F_{2^m}\rightarrow\F_{2^m}.
\end{split} \]
It was shown in Zheng et al. \cite{zheng} that $F(x)$ is the only vectorial function of the form $x^{2^e} h(\Tr_{2^n/2^m}(x))$ with $2^n-2^m$ bent components.

We first relate the Walsh transform and the nonlinearity of $F(x)$ to the Walsh transform of $H(x)$:
\begin{theorem}\label{th:wt}
For any $a\in\F_{2^m}^*$,
\begin{equation}\label{eq:wt}
  W_{F_a}(\omega)=\begin{cases}
  	 2^mW_{H_a}(\omega), &\text{if}\ \omega\in\F_{2^m},\\ 0, &\text{if}\ \omega\in\F_{2^n}-\F_{2^m}.
  \end{cases}
\end{equation}
Consequently,
\begin{equation}  \label{eq:nf} \N_F=2^{n-1}-
2^{m-1}\max_{a\in\F_{2^m}^*}\max_{\omega\in\F_{2^m}}|W_{H_a}(\omega)|.\end{equation}
\end{theorem}
\begin{proof} Let  $\bU$ be any set of representatives of cosets of $\F_{2^m}$ in $\F_{2^n}$, then every $x\in\F_{2^n}$ is uniquely written as  $x=u+y$ with $u\in\bU$ and $y\in\F_{2^m}$. Therefore
  \begin{align*}
    W_{F_a}(\omega) & =\sum_{x\in\F_{2^n}}(-1)^{\Tr_{2^n/2}(ax^{2^e}h(\Tr_{2^n/2^m}(x)))+\Tr_{2^n/2}(\omega x)} \\
     &=\sum_{u\in\bU}\sum_{y\in\F_{2^m}}(-1)^{\Tr_{2^n/2}(a(y+u)^{2^e}h(\Tr_{2^n/2^m}(u)))+\Tr_{2^n/2}(\omega (y+u))}  \\
     &= \sum_{u\in\bU}(-1)^{\Tr_{2^n/2}(au^{2^e}h(\Tr_{2^n/2^m}(u))+\omega u)}\sum_{y\in\F_{2^m}}(-1)^{\Tr_{2^n/2}(ay^{2^e}h(\Tr_{2^n/2^m}(u))+\omega y)}.
  \end{align*}
  Note that
\[ \sum_{y\in\F_{2^m}}(-1)^{\Tr_{2^n/2}(ay^{2^e}h(\Tr_{2^n/2^m}(u))+\omega y)}=\sum_{y\in\F_{2^m}}(-1)^{\Tr_{2^n/2}(\omega y)}
	=\begin{cases}
		0, \ & \omega\notin \F_{2^m},\\ 2^m,\ & \omega\in\F_{2^m}.
	\end{cases}\]
Then
\begin{align*}
  W_{F_a}(\omega)&=2^m\sum_{u\in\bU}(-1)^{\Tr_{2^n/2}(au^{2^e}h(\Tr_{2^n/2^m}(u))+\omega u)} \\
  &=2^m\sum_{u\in\bU}(-1)^{\Tr_{2^m/2}(a(u+u^{2^m})^{2^e}h(\Tr_{2^n/2^m}(u))+\omega (u+u^{2^m}))}.
\end{align*}
Note that $u\mapsto u+u^{2^m}, \bU\rightarrow \F_{2^m}$ is a bijection, substituting $u+u^{2^m}$ by $z$, then
\begin{equation*}
  W_{F_a}(\omega) =2^m\sum_{z\in\F_{2^m}}(-1)^{\Tr_{2^m/2}(az^{2^e}
  h(z)+\omega z)}=2^mW_{H_a}(\omega).
\end{equation*}
Thus we obtain Eq.~\eqref{eq:wt}. As a consequence, we have
\[\N_F=2^{n-1}-\frac{1}{2}\max_{a\in\F_{2^n}^*}\max_{\omega\in\F_{2^n}}|W_{F_a}(\omega)|=2^{n-1}-
2^{m-1}\max_{a\in\F_{2^m}^*}\max_{\omega\in\F_{2^m}}|W_{H_a}(\omega)|. \qedhere \]
\end{proof}

\subsection{Nonlinearity of monomial permutations} In this subsection, we assume
 \[ h(x)= x^{u-2^e},\quad H(x)= x^u\]
 where
  \[ \gcd(u,2^m-1)=\gcd(u-2^e,2^m-1)=1.\]
Then $h(x)$ and $H(x)$ are both permutations over $\F_{2^m}$. In this case,
\[W_{H_a}(\omega)=W_{H_1}(a^{-\frac{1}{u}}\omega).\]
Then \[\max_{a\in\F_{2^m}^*}\max_{\omega\in\F_{2^m}}|W_{H_a}(\omega)|=
\max_{a\in\F_{2^m}^*}\max_{\omega\in\F_{2^m}}
|W_{H_1}(a^{-\frac{1}{u}}\omega)|=\max_{\omega\in\F_{2^m}}
|W_{H_1}(\omega)|=:\L(u). \]
By Eq.~\eqref{eq:nf}, then
 \begin{equation}\label{eq:3} \N_F= 2^{n-1}-2^{m-1} \L(u). \end{equation}

 Note that if $u\equiv2^k\pmod{2^m-1}$, $k\geq0$, $H(x)$ is a linear function and $\N_F=0$. Hence, we assume $u\not\equiv2^k\pmod{2^m-1}$ in this subsection.
 There is a conjecture about $\L(u)$ and $\N_F$ (see \cite{dobb})
\begin{conjecture} $\L(u)\geq 2^{[\frac{n}{4}+1]}$, equivalently, Eq.\eqref{eq:nfbound} holds, i.e.
	\[
 \N_F\leq2^{n-1}-2^{\lfloor\frac{3n}{4}\rfloor}. \]
\end{conjecture}

\vskip 0.3cm
\noindent
\textbf{(I) The case of Niho exponents.}
\begin{theorem}\label{4} Assume $n=4t>1$ and  $u=s(2^t-1)+ 2^{e+1}$ such that $\gcd(2^{e}-s,2^t+1)=1$ and $0<s<2^t-1$. Then $\N_F\leq 2^{n-1}-2^{\frac{3n}{4}}$, i.e. the Conjecture is true in this case.
\end{theorem}
\begin{proof}  By the identity
	 \[ \L(u)=\L(2^{m-e-1}u), \]
replacing $(u,s,e)$ by $(2^{m-e-1}u\pmod{2^m-1}, 2^{m-e-1}s\pmod{2^t+1}, m-1)$, we may assume $u=s(2^t-1)+1$ such that $\gcd(1-2s,2^t+1)=1$.
	
Suppose $\F_{2^m}^*=\langle\beta\rangle$ and $\C=\langle\beta^{2^t-1}\rangle$. Then
	\begin{align*}
		W_{H_1}(\omega)=&\sum_{x\in\F_{2^m}}(-1)^{\Tr_{2^{m}/2}(x^{u}+\omega x)}\\
		=&1+\sum_{i=0}^{2^t}\sum_{z\in\F_{2^t}^*}(-1)^{\Tr_{2^{m}/2}((\beta^{(2^t-1)i}z)^{s(2^t-1)+1}+\omega \beta^{(2^t-1)i}z)}\\
		=&1+\sum_{y\in \C}\sum_{z\in\F_{2^t}^*}(-1)^{\Tr_{2^{t}/2}(\Tr_{2^m/2^t}(y^{1-2s}+\omega y)z)}.
	\end{align*}
	Set $N_{\omega}(y)=|\{y\in \C:\ 1+\omega y^{2s}+y^{2(2s-1)}+\omega^{2^t} y^{2s-2}=0\}|$. Then
	\begin{equation}\label{nh1}W_{H_1}(\omega)=(N_{\omega}(y)-1)2^t\geq -2^t. \end{equation}
	
	Obviously, $W_{H_1}(\omega)$ takes value $0$.
  Note that
  \[(\sum_{\omega\in\F_{2^m}}W_{H_1}(\omega))^2=
  \sum_{\omega\in\F_{2^m}}W^2_{H_1}(\omega).\]
  There must exist $\omega_1\in\F_{2^m}^*$ and  $\omega_2\in\F_{2^m}^*$ such that $W_{H_1}(\omega_1)>0$ and $W_{H_1}(\omega_2)<0$ respectively. By Eq.~\eqref{nh1}, we have $W_{H_1}(\omega_2)=-2^{t}$.

Suppose for some $\omega\in\F_{2^m}^*$, $W_{H_1}(\omega)=A$ with $2^t\mid A$ and $A\notin\{0,-2^t\}$. Set
\[L_\omega(y):=1+\omega y^{2s}+y^{2(2s-1)}+\omega^{2^t} y^{2s-2}.\]

 Note that for $\omega\in\F_{2^t}$, $L_\omega(1)=0$ and $L_\omega(y)=0$ if and only if $L_{\omega}(y^{2^t})=0$. Thus $N_{\omega}(y)=1+2\Delta$ for some $\Delta$. By Eq.~\eqref{nh1}, we have $A=2^{t+1}\Delta$. Next we prove $\Delta$ is $2$-power. Note that
\[\sum_{\omega\in\F_{2^t}}W_{H_1}(\omega)=\sum_{x\in\F_{2^m}}
(-1)^{\Tr_{2^m/2}(x^u)}\sum_{\omega\in\F_{2^t}}(-1)^{\Tr_{2^m/2}(\omega x)}=2^m.\]

On the left of the above sum, for $\omega\in\F_{2^t}$, if $W_{H_1}(\omega)\neq0$, then $2^{t+1}\mid W_{H_1}(\omega)$. Therefore, there must exist an integer $L$ such that $2^m=LA$, this means $A=2^{l}$, $l\geq t+1$. Then $\L(u)\geq 2^{t+1}$, and by Eq.~\eqref{eq:3}, we have \[\N_F\leq 2^{n-1}-2^{m+t}=2^{n-1}-2^{\frac{3n}{4}}.\qedhere\]
\end{proof}
\begin{remark} Helleseth \cite{Helle76} showed that $W_{H_1}(\omega)$ takes at least three values for $\omega\in\F_{2^m}^*$. The above proof is almost identical to that of \cite[Theorem~2]{cha04}.
\end{remark}
\begin{example} Suppose $4\mid t$. Take $(s,e)=(1, t-1)$. Then $u=2^{t+1}-1$ and $h(x)= x^{2^t+2^{t-1}-1}$.   In this case $W_{H_1}(\omega)$ takes values in $\{0, \pm 2^t, 2^{t+1}\}$ (see \cite[Theorem~2.2]{Helle76}) and $\N_F= 2^{n-1}-2^{\frac{3n}{4}}$, in this case $F(x)$ is not plateau.
\end{example}

\begin{example} Suppose $t=2t_1$ and $4\nmid (t_1-1)$. Take $(s,e)=(2^{t_1}-1, t_1)$. Then $u=2^{3t_1}-2^{2t_1}+2^{t_1}+1$ and $h(x)= x^{2^{3t_1}- 2^{2t_1}+1}$. In this case,  $W_{H_1}(\omega)$ takes values in $\{0, \pm 2^t, 2^{3t_1}\}$ (see \cite[Theorem~2.3]{Helle76}) and $\N_F= 2^{n-1}-2^{\frac{7n}{8}-1}$, in this case $F(x)$ is also not plateau.
\end{example}

\vskip 0.3cm
\noindent
\textbf{(II) The case of three-valued Walsh transform.}

Assume $\omega\in \F_{2^m}\mapsto W_{H_1}(\omega)$ is a three-value function.
The following proposition in \cite{au15} reveals more information about the explicit values of $W_{H_1}(\omega)$.

\begin{proposition}\label{three}
  If $W_{H_1}(\omega)$ is a three-value function with values $0$ and $\pm A$, then $|A|=2^k$ for some integer $k$. Furthermore, let $R$ denote the set of roots of $(x+1)^u-x^u-1$ in $\F_{2^m}$, then $2^{\frac{m}{2}}<\sqrt{2^mR}=|A|<2^m$.
\end{proposition}
Set $h_u(x):=(x+1)^u-x^u-1$. Then $h_u(0)=h_u(1)=0$. If $2\mid m$ and $3\nmid u$, then every element of $\F_{2^2}$ is a root of $h_u(x)$ in $\F_{2^m}$.
By Proposition~\ref{three}, we have
\begin{theorem}\label{5}
 With the same notations as Proposition~\ref{three}, then
\[\N_F=2^{n-1}-2^{m-1}\sqrt{2^mR}\leq\begin{cases}
   2^{n-1}-2^{m-1}\sqrt{2^{m+2}}, & \mbox{if } 2\mid m\ \text{and}\ 3\nmid u, \\
   2^{n-1}-2^{m-1}\sqrt{2^{m+1}}, & \mbox{if } 2\nmid m\ \text{or}\ 3\mid u.
 \end{cases}\]
\end{theorem}
Note that the assumptions imply that $F(x)$ is plateau, hence
the  Conjecture holds by Lemma~\ref{lemma:anb}, however, Theorem~\ref{5} gives a better bound for $\N_F$.
\begin{example}
 We list known three-value exponential $u$ in Table~1  given by \cite{Go68,au15}. One can see the upper bound is attained for the last five cases.
\begin{table}[!htbp]
	\setlength{\abovecaptionskip}{0.cm}
	\setlength{\belowcaptionskip}{-0.cm}
	\caption{Nonlinearity of three-values functions}
	\label{table:t1}
	\setlength{\tabcolsep}{0.5mm}{
		\begin{tabular}{|c|c|c|c|c|}
			\hline
			\multirow{1}{*}{$u$} &\multirow{1}{*}{Constraints} & \multirow{1}{*}{$h(x)$}  & \multirow{1}{*}{$W_{H_1}(\omega)$} & \multirow{1}{*}{Nonlinearity} \\
			\hline
			\multirow{2}{*}{$2^e+1$} &  \multirow{1}{*}{$v_2(e)\geq v_2(m)$}& \multirow{2}{*}{$x$} & \multirow{2}{*}{$0,\pm \sqrt{2^{d+m}}$} &\multirow{2}{*}{$2^{n-1}-2^{m-1}\sqrt{2^{d+m}}$}\\
			& $\gcd(m,e)=d$& & &\\
			\hline
			\multirow{2}{*}{$2^{2e}-2^e+1$} & \multirow{1}{*}{$v_2(e)>v_2(m)$,} & \multirow{2}{*}{$x^{1-2^e}$} & \multirow{2}{*}{$0,\pm \sqrt{2^{m+1}}$}
			& \multirow{2}{*}{$2^{n-1}-2^{m-1}\sqrt{2^{m+1}}$}\\
			& $\gcd(m,e)=1$ & & &\\
			\hline
			\multirow{2}{*}{$2^{\frac{m}{2}}+2^{\frac{m+2}{4}}+1$} & \multirow{1}{*}{$v_2(m)=1$,} & \multirow{2}{*}{$x^{1+2^{\frac{m+2}{4}}}$} & \multirow{2}{*}{$0,\pm 2\sqrt{2^{m}}$}
			& \multirow{2}{*}{$2^{n-1}-2^{m}\sqrt{2^{m}}$}\\
			& $v_2(\gcd(m,\frac{m+2}{4}))=1$& & &\\
			\hline
			\multirow{2}{*}{$2^{\frac{m+2}{4}}+3$} & \multirow{1}{*}{$v_2(m)=1$} & \multirow{2}{*}{$x^{1+2^{\frac{m+2}{4}}}$} &  \multirow{2}{*}{$0,\pm 2\sqrt{2^{m}}$}    &\multirow{2}{*}{$2^{n-1}-2^{m}\sqrt{2^{m}}$}\\
			& $v_2(\gcd(m,\frac{m+2}{4}))=1$& &  &\\
			\hline
			\multirow{2}{*}{$2^{\frac{m-1}{4}}+3$} &\multirow{2}{*}{$v_2(m)=0$} & \multirow{2}{*}{$x^{2^{\frac{m-1}{2}}+1}$ }&\multirow{2}{*}{$0,\pm \sqrt{2^{m+1}}$}   & \multirow{2}{*}{$2^{n-1}-2^{m-1}\sqrt{2^{m+1}}$}\\
			& & &  &\\
			\hline
			\multirow{2}{*}{$2^{2e}+2^e-1$} &\multirow{1}{*}{$v_2(m)=0$} & \multirow{2}{*}{$x^{2^{e}-1}$ }&\multirow{2}{*}{$0,\pm \sqrt{2^{m+1}}$}   & \multirow{2}{*}{$2^{n-1}-2^{m-1}\sqrt{2^{m+1}}$}\\
			& $m\mid(4e+1)$& & & \\
			\hline
	\end{tabular}}
\end{table}

\end{example}

\subsection{Nonlinearity of linear permutation} Assume $h$ is a linear permutation over $\F_{2^m}$. Then $H_a(x)=\Tr_{2^m/2}(ax^{2^e}h(x)):\F_{2^m}\rightarrow\F_2$ is a quadratic Boolean function. We let
 \begin{equation}
 	r_a:=\rank(H_a(x)),\quad r:= \min \{r_a \mid a\in\F_{2^m}^*\ \text{and}\ 2\mid r_a\}.
 \end{equation}
The following result in \cite[Theorem~1]{Hufeng}  gives the explicit value of $W_{H_a}(\omega)$.

\begin{proposition} \label{prop:13}	If $r_a$ is odd, $W_{H_a}(\omega)=0$. If $r_a$ is even,  $W_{H_a}(\omega) \in \{0, \pm2^{m-\frac{r_a}{2}}\}$, and $W_{H_a}(\omega) $ over all the $a$'s   with the same $r_a$ are not identically zero .	
\end{proposition}

By Proposition~\ref{prop:13}, we know that the quadratic vectorial function $H(x)=x^{2^e}h(x)$ is plateaued ($k=m-r_a$ depending on $a$). Note that $2\mid r\leq r_a\leq m$, then we have
\begin{theorem}\label{cor}
	Suppose $h(x)$ is a linear permutation over $\F_{2^m}$.  Then
	\[\N_F=2^{n-1}-2^{n-\frac{r}{2}-1} \leq 2^{n-1}-2^{[\frac{3n}{4}]}, \]
and the last inequality  is tight if only if
	\begin{enumerate}[1]
		\item $2\mid m$ and all even $r_a$ take values $m$ and $m-2$, or
		\item $2\nmid m$ and all even $r_a=m-1$.
	\end{enumerate}	
\end{theorem}

Set $h(x)=x$ and $\gcd(e,m)=d$. In \cite{cou}, we know that $r_a=m-d$ for $a\in\F_{2^m}^*$ if $\frac{m}{d}$ is odd, and $r_a=m$ or $m-2d$ for $a\in\F_{2^m}^*$ if $\frac{m}{d}$ is even.
Thus we have
\begin{corollary}\label{cor2}
Suppose $F(x)=x^{2^e}\Tr_{2^n/2^m}(x)$ and $\gcd(m,e)=d$. Then $\N_F$ attains the upper bound in Eq.~\eqref{eq:nfbound} if and only if $d=1$.
\end{corollary}

\section{Construction of Vectorial Functions With Maximum Number of Bent Components}

The main goal of this section is to construct vectorial functions from $V$ of dimension $n$ to itself with maximal number of bent components. The construction is based on the property ($\bA$), which was introduced in \cite{zpeng} and has been used to construct vectorial functions with maximal number of bent components \cite{bap21,bap}.
\begin{definition}\label{definition:3}
Suppose $f: V\rightarrow \F_2$ and $\{u_1,u_2,\ldots,u_k\}\subseteq V$ for $2\leq k\leq n$.  We say that $(f;u_1,\cdots, u_k)$ satisfies Condition $\bA $ if $f(x)$ is a bent function and for any $(w_1,w_2,\cdots,w_k)\in\F^k_2$,
\begin{equation*}f^*(x+\sum\limits_{i=1}^{k}w_iu_i)=f^*(x)+\sum\limits_{i=1}^{k}w_iD_{u_i}f^*(x) .\end{equation*}

\end{definition}

Let $V=\F_{2^n}$, $g:V\rightarrow W$ and $G:V\rightarrow\F_{2^m}$ such that $G_\beta^*$ satisfies $\bA$, $\beta\in\F_{2^m}^*$, and let $\{u_1,u_2,\ldots,u_k\}\subseteq V$ be linearly independent over $\F_2$ and
\[F(x)=G(x)+g(x).\]
\begin{itemize}
  \item Take $W=\F_2$.   Zheng et al. \cite{zkan} proved that $F(x)$ is a vectorial bent function.
\item Take $W=V$ and $g=\gamma h(\Tr_{2^n/2}(u_1x),\ldots,\Tr_{2^n/2}(u_kx))$, where $\gamma\in\F_{2^n}-\F_{2^m}$, $h(x):\F_2^k\rightarrow\F_{2^k}$, $k\mid m$. Bapi\'{c} and Pasalic \cite{bap21} showed that $F(x)$ has $2^n-2^m$ bent components.
    \item Take $W=V$ and $g=\gamma h(\Tr_{2^n/2}(u_1x),\ldots,\Tr_{2^n/2}(u_kx))$, where $\gamma\in\F_{2^n}-\F_{2^m}$, $h(x):\F_2^k\rightarrow\F_{2^s}$, $k\leq m$ and $s\mid m$. Bapi\'{c} et al. \cite{bap} showed that  $F(x)$ has $2^n-2^m$ bent components and contains some bent components which are not in Maiorana-McFarland construction.
\end{itemize}

In this section, we will give a slightly revision of \cite{zpeng,zkan} and obtain some vectorial functions with good properties.

\subsection{Construction via the Niho quadratic function}

From now on in this subsection, let $V=\F_{2^n}$ and $G(x)=x^{2^m+1}$. For $\beta\in\F_{2^n}-\F_{2^m}$, let
 \[ \gamma=\beta+\beta^{2^m}\in \F_{2^m}^*. \]
The component function of $G$ at $\beta$ is the monomial Niho quadratic function
\begin{equation*}
G_\beta: x\in\F_{2^n}\mapsto\Tr_{2^n/2}(\beta x^{2^m+1}).
\end{equation*}
It is a bent function (see \cite{book}) and its dual $G_\beta^*$ is given by
\begin{equation}\label{eq:dniho}
G_\beta^*(x)=\Tr_{2^m/2}(\gamma^{-1}x^{2^m+1})+1.
\end{equation}

We first show that the function $G_\beta(x)$ satisfies Condition $\bA $ when $u_1,u_2,\cdots,u_k$ are appropriately chosen.
\begin{lemma}\label{lemma:a1}
	Suppose $\beta\in\F_{2^n}-\F_{2^m}$, $3\leq k\leq m$ and $\{u_1,u_2,\cdots,u_k\}\subseteq \F_{2^m}$ such that $\Tr_ {2^m/2}(u_1u_i)=0$ for all $2\leq i\leq k$.   Then $(G_\beta(x);\beta u_1,\cdots, u_k)$  satisfies Condition $\bA $ and
\begin{align*}
	D_{\beta u_1}G^*_\beta(x) & =\Tr_{2^n/2}(\gamma^{-1}u_1\beta^{2^m} x)+\Tr_{2^m/2}(\gamma^{-1}\beta^{2^m+1} u_1^{2}),
	\end{align*}
and for $2\leq i\leq k$,
	\begin{align*}
	D_{u_i}G^*_\beta(x) & =\Tr_{2^n/2}(\gamma^{-1}x u_i)+\Tr_{2^m/2}(\gamma^{-1} u_i^{2}).
	\end{align*}
\end{lemma}
\begin{proof} By Eq.~\eqref{eq:dniho},  the derivative of $G^*_\beta(x)$ in the direction of $u_j\in\F_{2^m}$ is
	\begin{align*}
	D_{u_j}G^*_\beta(x) &=\Tr_{2^m/2}(\gamma^{-1}x^{2^m+1})+1+\Tr_{2^m/2}(\gamma^{-1}(x+u_j)^{2^m+1})+1  \\
	& =\Tr_{2^n/2}(\gamma^{-1}u_jx)+\Tr_{2^m/2}(\gamma^{-1}u_j^{2}).
	\end{align*}

Since $u_i\in\F_{2^m}$  and $\Tr_{2^m/2}(u_1 u_j)=0$,
we have $\gamma^{-1}u_j u_i\in\F_{2^m}$. Then	the second order derivatives in the direction of $(u_i,u_j)$ and $(\beta u_1,u_j)$ are
	\begin{equation*}D_{u_i}D_{u_j}G^*_{\beta}(x)=
\Tr_{2^n/2}(\gamma^{-1}u_j
u_i)=0,\ 2\leq i<j\leq k\leq m,\end{equation*}and
\begin{equation*}D_{\beta u_1}D_{u_j}G^*_{\beta}(x)=
\Tr_{2^n/2}(\gamma^{-1}u_j\beta
u_1)=\Tr_{2^m/2}(u_ju_1)=0,\  2\leq j\leq k.\qedhere
\end{equation*}
\end{proof}

By Lemma~\ref{lemma:a1}, our first construction of vectorial functions
with maximal number of bent components is the following result:

\begin{theorem}\label{theorem:nvec} Let $3\leq k\leq m$ and $\{u_1, u_2,\cdots,u_k\}\subseteq\F_{2^m}$ satisfy $\Tr_{2^m/2}(u_1 u_j)=0$ for $j\geq 2$. Then  for any reduced polynomial $R(X_2,\cdots, X_k)$ over $\F_2$,	
	\[F(x)=x^{2^m+1}+u_1 xR(\Tr_{2^n/2}(u_2x),\Tr_{2^n/2}(u_3x),\cdots,
\Tr_{2^n/2}(u_kx)),\]
has $2^n-2^m$ bent components. More precisely,  for $\beta\in\F_{2^n}-\F_{2^m}$, $F_\beta$ is bent and
	\begin{equation*}
	F_\beta^*(x)=G^*_\beta(x)+D_{\beta u_1}G^*_{\beta}(x)R(D_{u_2}G^*_{\beta}(x),\cdots,D_{u_k}G^*_{\beta}(x)).
	\end{equation*}
\end{theorem}

In Theorem~\ref{theorem:nvec}, if we take $k=m$ and $\{u_i:1\leq i\leq m\}$ to be an orthogonal basis of $\F_{2^m}$ over $\F_2$ under the inner product $\langle x,y\rangle=\Tr_{2^m/2}(xy)$, then Condition $\bA$ holds, and what's more, we have the following result.
\begin{corollary}Let $R(X_2,  \ldots, X_m) = X_2\cdots X_m$.
Then the function
\[F(x)=x^{2^m+1}+u_1 x\prod_{i=2}^m \Tr_{2^n/2}(u_ix)\]
has maximal algebraic degree $m$ and maximal number of bent components $2^n-2^m$.\end{corollary}

In particular,
for $k=2$, i.e., $R(X_2, \ldots, X_m) = X_2$, we can take $u_1, u_2\in\F_{2^n}$ such that $(G_\beta;\beta u_1,u_2)$ satisfies Condition $\bf A$. Thus we have
\begin{theorem}\label{th:t}
Suppose $u_1, u_2\in\F_{2^n}$ such that $u_1u_2^{2^m}\in\F_{2^m}$ and $\Tr_{2^m/2}(u_1u_2^{2^m})=0$. Then
  \[F(x)=x^{2^m+1}+u_1x\Tr_{2^n/2}(u_2 x)\]
	has $2^n-2^m$ bent components: for $\beta\in \F_{2^n}-\F_{2^m}$,  $F_\beta$ is bent and
\begin{align*}F^*_\beta(x)=&\Tr_{2^m/2}(\lambda^{-1}x^{2^m+1})+1
+\bigr(\Tr_{2^n/2}(\lambda^{-1}(\beta u_1)^{2^m}x)
+\Tr_{2^m/2}(\lambda^{-1}(\beta u_1)^{2^m+1})\bigl)\\
&\times \bigr(\Tr_{2^n/2}(\lambda^{-1}u^{2^m}_2x)
+\Tr_{2^m/2}(\lambda^{-1}u_2^{2^m+1})\bigl).\end{align*}
\end{theorem}
\begin{proof}
  The proof can be obtained directly from Lemma~\ref{lemma:a1}, so we omit it here.
\end{proof}
In Theorem~\ref{th:t}, $W_{F_\beta}(\omega)=\pm2^{m}$ for $\beta\in\F_{2^n}-\F_{2^m}$.  For $\beta\in\F_{2^m}^*$, we have
\begin{theorem}\label{rt}
With the same notations as Theorem~\ref{th:t}. For $\beta\in\F_{2^m}^*$, we have
  \begin{enumerate}[1]
    \item If $\beta u_1,u_2$ are linearly independent over $\F_2$ for some $\beta\in\F_{2^m}^*$, then $W_{F_\beta}(\omega)=\pm2^{n-1}$ if $\omega\in\{0,\beta u_1,u_2,\beta u_1+u_2\}$. Otherwise, $W_{F_\beta}(\omega)=0$.
    \item  If $\beta u_1,u_2$ are linearly dependent over $\F_2$ for some $\beta\in\F_{2^m}^*$, then $W_{F_\beta}(\omega)=2^{n}$ if $\omega=\beta u_1=u_2$. Otherwise, $W_{F_\beta}(\omega)=0$.
  \end{enumerate}
Furthermore, for the first case $\N_F=2^{n-2}$.
\end{theorem}
\begin{proof}
We have
\[W_{F_\beta}(\omega)=\sum_{x\in\F_{2^n}}(-1)^{\Tr_{2^n/2}(\beta u_1x)\Tr_{2^n/2}(u_2x)+\Tr_{2^n/2}(\omega x)}.\]
For some $\beta\in\F_{2^m}^*$, if $\beta u_1,u_2$ are linearly dependent over $\F_2$, then $F_{\beta}(x)$ is linear, the result is trivial.

If $\beta u_1,u_2$ are linearly independent over $\F_2$, then
\begin{equation}\label{q}W_{F_\beta}(\omega)=\sum_{\substack{x\in \F_{2^m},\\[2pt] \Tr_{2^n/2}(u_2x)=0}}(-1)^{\Tr_{2^n/2}(\omega x)}+\sum_{\substack{x\in \F_{2^m},\\[2pt] \Tr_{2^n/2}(u_2x)=1}}(-1)^{\Tr_{2^n/2}((\beta u_1+\omega) x)}.\end{equation}
Obviously, Eq.~\eqref{q} equals $0$ if $\omega\notin\{0,\beta u_1,u_2,\beta u_1+u_2\}$,  equals $2^{n-1}$ if $\omega\in \{0,\beta u_1,u_2\}$,  and equals $-2^{n-1}$ if $\omega=\beta u_1+u_2$.
\end{proof}

\begin{remark}\label{remr3}In Theorem~\ref{rt} (1), we know $F_\beta(x)=\Tr_{2^n/2}(\beta u_1x)\Tr_{2^n/2}(u_2x)$. Note that the non-degeneracy of $\Tr_{2^n/2}$ means that, for all $y\in\F_{2^n}$, the equation
\[F_\beta(x+y)+F_\beta(x)+F_\beta(y)=\Tr_{2^n/2}\left(y(u_2\Tr_{2^n/2}(\beta u_1x)+\beta u_1\Tr_{2^n/2}(u_2x))\right)=0\]if and only if
$u_2\Tr_{2^n/2}(\beta u_1x)+\beta u_1\Tr_{2^n/2}(u_2x)=0$ for $x\in\F_{2^n}$, i.e.,
 \[\begin{cases}
                      \Tr_{2^n/2}(u_2x)=0, \\
                      \Tr_{2^n/2}(\beta u_1x)=0.
                    \end{cases}\]The number of solutions of the above system of equations is $2^{n-2}$, which means $\rank(F_\beta(x))=2.$ By Theorem~\ref{cor}, $r_\beta=1$ and $\N_F=2^{n-1}-2^{n-2}=2^{n-2}$.\end{remark}

Note that the function $F(x)$ given by Theorem~\ref{th:t} is of the form $x\ell(x)$, thus is a solution of a problem proposed by Pott et al. \cite{5}. We now show it is not equivalent to the known functions  of the form $x\ell(x)$.

Recall for a vectorial function $F$ and $a, b\in V$, $\delta_F(a,b):=|\{x\in\F_{2^n}:\ F(x+a)+F(x)=b\}|$. The differential spectrum of $F$ is
 \[\{\delta_F(a,b):\ a\in\F_{2^n}^*,\ b\in\F_{2^n}\}.\]
It was shown in \cite{ny} and \cite{5} respectively that
 \[\delta_{x^{2^m+1}}(a,b)\in\{0,2^m\}\ \text{and}\ \delta_{x^{2^i}(1+x^{2^m})}(a,b)\in\{0,2^{\gcd(i,m)},2^m\}\ (0<i<m).\]In addition, Anbar et al. \cite{anb} proved that \[\delta_{x^{2^e}\Tr_{2^n/2^m}(h(x))}(a,b)\in\{0,2^m\}\ \text{if}\ a\in\F_{2^m}^*, \text{and}\ \in\{0,2^s\}\ \text{if}\ a\in\F_{2^n}-\F_{2^m},\] where $s$ is the dimension of the solution space (in $\F_{2^m}$) of $x^{2^e}h(\Tr_{2^n/2^m}(a))+\Tr_{2^n/2^m}(a)^{2^e}h(x)=0$.
Then the inequivalence of our function in Theorem~\ref{th:t} to the above functions follows from
\begin{theorem}\label{theorem:diff}
  Suppose $u_1, u_2\in\F_{2^n}-\F_{2^m}$, $u_1u_2^{2^m}\in\F_{2^m}$ and $\Tr_{2^m/2}(u_1u_2^{2^m})=0$. Then the differential spectrum of  $F(x)=x^{2^m+1}+u_1x\Tr_{2^n/2}(u_2 x)$ is given by
  \[\delta_{F}(a,b)\in \begin{cases}
                      \{0,2\}, & \mbox{if }\ \Tr_{2^n/2}(u_2a)=1, \\
                      \{0,2^{m-1},2^m\}, & \mbox{if }\ \Tr_{2^n/2}(u_2a)=0.
                    \end{cases}\]
\end{theorem}
\begin{proof}
  We have
  \begin{align*}
    F(x+a)&+F(x) \\
     &=x^{2^m}a+xa^{2^m}+a^{2^m+1}+u_1x\Tr_{2^n/2}(u_2a)+u_1a\Tr_{2^n/2}(u_2 (x+a)).
  \end{align*}

  Notice that if $x$ is a solution of $F(x+a)+F(x)=b$, so is $x+a$.

 (A) Assume $\Tr_{2^n/2}(u_2a)=1$. The equation $F(x+a)+F(x)=b$ is reduced to
  \begin{equation*}
    x^{2^m}a+xa^{2^m}+a^{2^m+1}+u_1x+u_1a\Tr_{2^n/2}(u_2x)+u_1a=b,
  \end{equation*}
and then to one of the following two systems of equations:
\[ \begin{cases} x^{2^m}a+xa^{2^m}+u_1x=b+a^{2^m+1}+u_1a,\\ \Tr_{2^n/2}(u_2 x)=0;
\end{cases}\  \begin{cases} x^{2^m}a+xa^{2^m}+u_1x=b+a^{2^m+1},\\ \Tr_{2^n/2}(u_2 x)=1.
\end{cases}
\]

We claim that $x^{2^m}a+xa^{2^m}+u_1x$ is a permutation over $\F_{2^n}$. Then $\delta_F(a,b)\in \{0,2\}$ follows from the claim immediately.

For $x,y\in\F_{2^n}$, let
  \[x^{2^m}a+xa^{2^m}+u_1x=y^{2^m}a+ya^{2^m}+u_1y.\]

  Set $z=\Tr_{2^n/2^m}(xa^{2^m})-\Tr_{2^n/2^m}(ya^{2^m})\in\F_{2^m}$, then $y=x+u_1^{-1}z$ and
  \begin{align*}
    z&=\Tr_{2^n/2^m}(xa^{2^m}-ya^{2^m})
    =-\Tr_{2^n/2^m}(a^{2^m}u_1^{-1}z)=-z\Tr_{2^n/2^m}(a^{2^m}u_1^{-1})\\
    &\Rightarrow z(1+\Tr_{2^n/2^m}(a^{2^m}u_1^{-1}))=0.
  \end{align*}

  Suppose $\Tr_{2^n/2^m}(a^{2^m}u_1^{-1})=1$. Notice that $u^{2^m}_1u_2\in\F_{2^m}^*$ and $au_1^{- 2^m}=\frac{au_2}{u_1^{2^m}u_2}$, thus
  \begin{align*}\Tr_{2^n/2^m}(a^{2^m}u_1^{-1})&=\Tr_{2^n/2^m}(au_1^{-2^m})=\Tr_{2^n/2^m}(\frac{au_2}{u_1^{2^m}u_2})
  =\frac{\Tr_{2^n/2^m}(au_2)}{u_1^{2^m}u_2}=1 \\
  &\Rightarrow \Tr_{2^n/2^m}(au_2)=u_1^{2^m}u_2 .
  \end{align*}

  Since $\Tr_{2^n/2}(u_2a)=1$, we have $\Tr_{2^n/2}(u_2a)=\Tr_{2^m/2}(u_1^{2^m}u_2)=1$,
   which is a contradiction to the assumption $\Tr_{2^m/2}(u_1^{2^m}u_2)=0$. Thus $z=0$ and $x^{2^m}a+xa^{2^m}+u_1x$ is a linear permutation over $\F_{2^n}$.

 \smallskip

 (B) Assume $\Tr_{2^n/2}(u_2a)=0$. The equation $F(x+a)+F(x)=b$ is reduced to
  \begin{equation}\label{eq:diff3}
    x^{2^m}a+xa^{2^m}+a^{2^m+1}+u_1a\Tr_{2^n/2}(u_2x)=b.
  \end{equation}
  Assume that $x,y$ are two solutions of~\eqref{eq:diff3}. Then
  \begin{align*}
    x^{2^m}a+xa^{2^m}+a^{2^m+1}+u_1a\Tr_{2^n/2}(u_2x)&=b ,\\
   y^{2^m}a+ya^{2^m}+a^{2^m+1}+u_1a\Tr_{2^n/2}(u_2y)&=b,
  \end{align*}
which means that $z=x+y$ is a solution of
  \begin{equation*}
    z^{2^m}a+z a^{2^m}+u_1a\Tr_{2^n/2}(u_2 z)=0.
  \end{equation*}
or equivalently,
\[ \begin{cases} z^{2^m}a+z a^{2^m}=0,\\ \Tr_{2^n/2}(u_2 z)=0;
\end{cases}\quad\text{or}\quad   \begin{cases} z^{2^m}a+z a^{2^m}=u_1 a,\\ \Tr_{2^n/2}(u_2 z)=1.
\end{cases}
\]
Thus $\delta_F(a,b)=0$  or the number of solutions of  these two systems of equations.

Let $X_u=\{x\in \F_{2^n}: \Tr_{2^n/2}(ux)=0\}$.
The zero set of the first system of equations is the $\F_2$-vector space $a^{-2^m}\F_{2^m} \cap X_{u_2}$. Note that  $\dim_{\F_2} a^{-2^m}\F_{2^m}=m$ and $\dim_{\F_2} X_{u_2}=n-1$,  $a^{-2^m}\F_{2^m} \cap X_{u_2}$ must be of dimension either $m-1$ or $m$.

For the second system, note that  $z^{2^m}a+z a^{2^m}\in \F_{2^m}$, we must have $u_1 a\in \F_{2^m}$. Hence $u_2a^{-2^m}=\frac{u_1^{2^m}u_2}{(u_1a)^{2^m}}\in\F_{2^m}$. The solution of $z^{2^m}a+z a^{2^m}=u_1 a$ is  $z=\frac{u_1a}{a^{2^m}(1+\xi)}$ with $\xi^{2^m+1}=1$.
Then
\begin{align*}\Tr_{2^n/2}(u_2z)&=\Tr_{2^n/2}(\frac{u_2 u_1a}{a^{2^m}(1+\xi)})=\Tr_{2^n/2}(\frac{u_2 a^{-2^m} u_1a}{1+\xi})\\
	& =\Tr_{2^m/2}(u_2 a^{-2^m} u_1a)=\Tr_{2^m/2}(u_2u_1^{2^m})=0. \end{align*}

Hence the second system has no zeros at all.
\end{proof}

\subsection{Construction via the Maiorana-MacFarland class} From this subsection, we let $V=\F_{2^m}\times\F_{2^m}$ and the corresponding inner product be
\[\langle(y_1,z_1),(y_2,z_2)\rangle=\Tr_{2^m/2}(y_1y_2)+\Tr_{2^m/2}(z_1z_2).\]

Suppose $\phi$ is a permutation  of $\F_{2^m}$ and $g:\F_{2^m}\rightarrow\F_{2^m}$, and let $G$ be the associated  map  defined by
 \begin{align*}
	G:\F_{2^m}\times\F_{2^m}&\longrightarrow\F_{2^m}\times\F_{2^m} \notag \\
	(y,z) & \longmapsto (y\phi(z),g(z)).
\end{align*}
Then $G$ has maximal number of bent components: for $(a,b)\in\F_{2^m}^*\times\F_{2^m}$, the
component function
\[G_{a,b}(y,z)=\Tr_{2^m/2}(ay\phi(z)+bg(z))\]
at $(a,b)$ is a bent function. A bent function which has the form $G_{a,b}(y,z)$ is the so-called Maiorana-MacFarland class.

Let $g(z)=z$ and $\phi$ be an automorphism of $\F_{2^m}$ from now on. Then
\begin{equation*} G_{a,b}(y,z)=\Tr_{2^m/2}(ay\phi(z)+bz),\end{equation*}
and its dual
\begin{equation*}
G_{a,b}^*(y,z)=\Tr_{2^m/2}((z+b)\phi^{-1}(a^{-1}y)).
\end{equation*}

Our second construction of vectorial functions with maximal number of bent components is the following result:

\begin{theorem}\label{11} Let $2\leq k\leq m$,  $\phi$ and $G$ be given as above.
 Suppose $u_1=(u_{1,1}, 0)$ and choose    $u_i=(u_{i,1},u_{i,2})$ for $2\leq i\leq k$ such that \[ \Tr_{2^m/2}(\phi^{-1}(u_{1,1})u_{i,2}) =0\ \text{ and  }\ u_{i,1}=\phi(u_{i,2}). \]  Then  for any reduced polynomial $R(X_2,\cdots X_k)$ over $
 \F_2$, the vectorial function
 \begin{align*}
F(y,z)=(y\phi(z),z)
+(u_{1,1}y,0)R(\Tr_{2^m/2}(u_{2,1}y+u_{2,2}z),
	\ldots,
\Tr_{2^m/2}(u_{k,1}y+u_{k,2}z))
\end{align*}
has $2^n-2^m$ bent components: for any $(a,b)\in\F_{2^m}^*\times\F_{2^m}$,
\begin{align*}
F_{a,b}(y,z) &= \langle(a,b),F(y,z)\rangle =\Tr_{2^m/2}(ay\phi(z)+bz) \\ &+\Tr_{2^m/2}(au_{1,1}y)R(\Tr_{2^m/2}(u_{2,1}y+u_{2,2}z),
	\ldots,\Tr_{2^m/2}(u_{k,1}y+u_{k,2}z))
\end{align*} is bent and
\[ F^*_{a,b}(y,z) =G_{a,b}^*(y,z)
+D_{u_{1,1}a,0}G_{a,b}^*(y,z)R(D_{u_{2}}G_{a,b}^*(y,z),
	\ldots, D_{u_k}G_{a,b}^*(y,z)).\]
\end{theorem}
\begin{proof} Since $\Tr_{2^m/2}(\phi^{-1}(u_{1,1})u_{i,2}) =0 $ and  $u_{i,1}=\phi(u_{i,2})$, then for $2\leq i\leq k$,
 \begin{align*} D_{u_i}D_{(au_{1,1},0)}G^*_{a,b}(y,z)&=
\Tr_{2^m/2}(\phi^{-1}(a^{-1})u_{i,2}\phi^{-1}(au_{1,1})
)\\
&=\Tr_{2^m/2}(u_{i,2}\phi^{-1}(u_{1,1}))=0,\end{align*}
and for $2\leq i<j\leq k$,
 \[ D_{u_j}D_{u_i}G^*_{a,b}(y,z)
 =\Tr_{2^m/2}\bigl(\phi^{-1}(a^{-1})(u_{j,2}\phi^{-1}(u_{i,1})
+u_{i,2}\phi^{-1}(u_{j,1}))\bigr)=0.\]
 Thus, $(G_{a,b};(au_{1,1},0), u_2,\ldots,u_k)$ satisfies Condition $\bA$ for $(a,b)\in\F_{2^m}^*\times\F_{2^m}$.
\end{proof}

Take $k=2$ and $R(X_2,\ldots,X_k)=X_2$. Then we have
\begin{theorem}
Suppose $u=(u_{1,1}, 0)$ and   $u_2=(u_{2,1},u_{2,2})$ such that $u_{2,1}=\phi(u_{2,2})$ and $ \Tr_{2^m/2}(\phi^{-1}(u_{1,1})u_{2,2}) =0$ . Then
\[F(y,z)=(y\phi(z),z)
+(u_{1,1}y,0)\Tr_{2^m/2}(u_{2,1}y+u_{2,2}z)\]
has $2^n-2^m$ bent components.
\end{theorem}

\begin{remark} The above vectorial function is not in the so-called
the complete Maiorana-MacFarland class.
	
In fact, by a result due to \cite{mes3},	a function $G:\F_{2^m}\times\F_{2^m}\rightarrow\F_2$ is in the complete Maiorana-MacFarland class  if and only if for any $w_1,w_2\in\F_{2^m}$,
	\begin{equation}\label{qq}D_{(w_1,0)}D_{(w_2,0)}G(y,z)\mid_{y=0}=0,\ \forall z\in\F_{2^m}.\end{equation}
It is enough to find two elements $a\in\F_{2^m}^*$ and $b\in\F_{2^m}$ such that Eq.~\eqref{qq} do not hold. For any $w_1,w_2\in\F_{2^m}$, we have
\begin{align}
D_{(w_1,0)}&D_{(w_2,0)}F_{a,b}(y,z)\notag \\
&=\Tr_{2^m/2}(au_{1,1}w_2)\Tr_{2^m/2}(u_{2,1}w_1)+\Tr_{2^m/2}(au_{1,1}w_1)\Tr_{2^m/2}(u_{2,1}w_2)\label{qa}
\end{align}
Note that for any $w_2\in\F_{2^m}$ Eq.~\eqref{qa} equals to 0 means that \begin{equation}\label{f}au_{1,1}\Tr_{2^m/2}(u_{2,1}w_1)+u_{2,1}\Tr_{2^m/2}(au_{1,1}w_1)=0,\ \text{for any}\ w_1\in\F_{2^m}. \end{equation}
Eq.~\eqref{f} is impossible for any $a\in\F_{2^m}^*$.
\end{remark}

\section{Binomial vectorial functions with maximal number of bent components}
Let $n=2m$. The main result of this section is
\begin{theorem}\label{theorem:bino}
	The binomial vectorial function $F(x)=x^{2^m+1}+x^{2^i+1}$ for $0\leq i\leq m-1$ on $\F_{2^n}$ has $2^n-2^m$ bent components if and only if $i=0$, i.e., $F(x)$ is affine equivalent to $x^{2^m+1}$.
\end{theorem}
\begin{remark}
	The special case of odd $m$ was proved by Zheng et al. \cite{zheng}.
\end{remark}

From now on, fix $i$ such that $0\leq i<m$, and let
\[ d=\gcd(m+i,2m)= \gcd(m+i,2i). \]

Let $F(x)=x^{2^m+1}+x^{2^i+1}$.
For $a\in\F_{2^n}$, the component function
$F_a(x)=\Tr_{2^n/2}(ax^{2^m+1}+ax^{2^i+1})$.  Let
\begin{equation} L_a(y):=a^{2^i}y^{2^{2i}}+(a+a^{2^m})^{2^i}y^{2^{m+i}}+ay. \label{eq:ln}\end{equation}

If $a\in \F_{2^m}$, then $F_a(x)=\Tr_{2^n/2}(ax^{2^i+1})$ and \eqref{eq:ln} is reduced to
 \begin{equation*} L_a(y):=a^{2^i}y^{2^{2i}}+ay. \end{equation*}

For any $y\in\F_{2^n}^*$, the derivative of $F_a(x)$ at direction $y$ is \begin{align*}D_yF_a(x)&=\Tr_{2^n/2}(a((x+y)^{2^m+1}+(x+y)^{2^i+1}))+\Tr_{2^n/2}(a(x^{2^m+1}+x^{2^i+1}))\\
	&=\Tr_{2^n/2}(x(ay^{2^i}+(a+a^{2^m})y^{2^m}+a^{2^{n-i}}y^{2^{n-i}}))= \Tr_{2^n/2}(x L_a(y)^{-2^i}).
\end{align*}

The root set of $L_a(y)$ in $\F_{2^n}$ forms an $\F_{2^d}$-vector space, hence the number of the roots of $L_a(y)$ in $\F_{2^n}$ is either $1$ or a power of $2^d$.

\begin{lemma}\label{lemma:bl}
Assume  $v_2(i)=v_2(m)$. For  $\xi\in \F_{2^d}$ such that $\xi^{2^{d/2}+1}=1$, let  $a=\frac{1}{1+\xi}$. Then $a\notin \F_{2^m}$ and $L_a(y)=0$  for any $y\in \F_{2^d}$.
\end{lemma}
\begin{proof} By $v_2(i)=v_2(m)$,  $d$ is even, $m=\frac{d}{2}\cdot m'$ and  $i=\frac{d}{2}\cdot i'$  with $m'$ and $i'$ both odd. Then
	\[ \xi^{2^m}=\xi^{2^i}=\xi^{-1}\ \Longrightarrow a^{2^m}=a^{2^i}=\xi a.\]
This means that $a\notin \F_{2^m}$ and
 \[L_a(y)=a^{2^i}y^{2^{2i}}+(a+a^{2^m})^{2^i}y^{2^{m+i}}+ay=
 a(\xi y^{2^{2i}}+(1+\xi) y^{2^{m+i}}+y). \]
For any $y_0\in \F_{2^d}=\F_{2^{2i}}\cap \F_{2^{m+i}}$, one has $ y^{2^{2i}}_0=  y^{2^{m+i}}_0$, hence $L_a(y_0)=0$.
\end{proof}

We need the following two general results.
\begin{lemma}\label{lem:b1}\cite[Theorem 5.30]{lidl}
Let $\chi'$ be a multiplicative character of $\F^*_{2^m}$ of order $2^d-1$. Then for any $(a,b)\in\F_{2^m}^*\times\F_{2^m}$,
 \[\sum_{x\in\F_{2^m}}(-1)^{\Tr_{2^m/2}(ax^{2^d-1}+b)}=(-1)^{\Tr_{2^m/2}(b)}
 \sum_{j=1}^{2^d-2}\overline{\chi'^j}(a)G(\chi'^j),\]
where  $\overline{\chi}$ and $G(\chi)$ are the  conjugate  and  the Gauss sum of $\chi$.
\end{lemma}

\begin{lemma}\label{lemma:gauss} Suppose $d<m$ is a factor of $m$. Let $\gcd(2^d-1,\frac{m}{d})=t$. Then the set  $N=\{y\in\F_{2^m}^*:\ \Tr_{2^m/2^d}(y^{2^d-1})=0\}$ has order
  \[|N| =\begin{cases}
  	\dfrac{2^m-2^d}{2^d}+\dfrac{(2^d-1)(-1)^{\frac{m}{d}-1}}{2^d}
  	\sum\limits_{\chi\in(\widehat{\F}_{2^d}^*)^{\frac{2^d-1}{t}}\backslash\{\chi_0\}}G(\chi)^{\frac{m}{d}}, &\text{if}\ t\neq1;\\ 2^{m-d}-1, \ &\text{if}\ t=1.
  \end{cases}\]
where  $\widehat{\F}_{2^d}^*$  is the set of the multiplicative characters of $\F_{2^d}^*$ and  $\chi_0$ is the trivial character. In particular, $N$ is non-empty.
\end{lemma}

\begin{proof} We have
  \begin{align}
    |N| &=\frac{1}{2^d}\sum_{v\in\F_{2^d}}\sum_{y\in\F^*_{2^m}}
    (-1)^{\Tr_{2^d/2}(v\Tr_{2^m/2^d}(y^{2^d-1}))}\notag\\
  &=\frac{2^m-1}{2^d}+\frac{1}{2^d}
\sum_{v\in\F^*_{2^d}}\sum_{y\in\F_{2^m}^*}
    (-1)^{\Tr_{2^m/2}(vy^{2^d-1})}.\label{eq:gu}
    \end{align}

Suppose $\F_{2^m}^*=\langle\beta\rangle$, then $\F_{2^m}^*=\bigcup_{i=0}^{\frac{2^m-1}{2^d-1}-1}\beta^i\F_{2^d}^*$. Note that $\gcd(\frac{2^m-1}{2^d-1},2^d-1)=\gcd(\frac{m}{d},2^d-1)=t$. If $t=1$, one has
 \begin{align*}
    |N|
  &=\frac{2^m-1}{2^d}+\frac{2^d-1}{2^d}
\sum_{v\in\F^*_{2^d}}\sum_{i=0}^{\frac{2^m-1}{2^d-1}-1}
    (-1)^{\Tr_{2^m/2}(v\beta^{i(2^d-1)})}\\
  &=\frac{2^m-1}{2^d}+\frac{2^d-1}{2^d}
\sum_{v\in\F^*_{2^d}}\sum_{i=0}^{\frac{2^m-1}{2^d-1}-1}
    (-1)^{\Tr_{2^m/2}(v\beta^{i})}\\
&=\frac{2^m-1}{2^d}+\frac{2^d-1}{2^d}
\sum_{v\in\F^*_{2^m}}
    (-1)^{\Tr_{2^m/2}(v)}=\frac{2^m-2^d}{2^d}\geq 1.
    \end{align*}

If $t\neq1$, suppose $\chi'$ is a multiplicative character of $\F_{2^m}^*$ of order $2^d-1$, then by Lemma~\ref{lem:b1} and Eq.~\eqref{eq:gu},
 \begin{align}
 |N|&=\frac{2^m-1}{2^d}+\frac{1}{2^d}
\sum_{v\in\F^*_{2^d}}\Bigl(\sum_{y\in\F_{2^m}}
    (-1)^{\Tr_{2^m/2}(vy^{2^d-1})}-1\Bigr)\notag\\
    &=\frac{2^m-1}{2^d}+\frac{1}{2^d}
\sum_{v\in\F^*_{2^d}} \Bigl( \sum_{j=1}^{2^d-2}
\overline{\chi'^j}(v)G(\chi'^j)-1 \Bigr)\notag\\
&=\frac{2^m-1}{2^d}+\frac{1}{2^d}
\sum_{v\in\F^*_{2^d}}\sum_{j=0}^{2^d-2}
\overline{\chi'^j}(v)G(\chi'^j).\label{eq:gu2}
 \end{align}

Suppose $\N$ is the norm mapping from $\F_{2^m}$ to $\F_{2^d}$. For $\chi\in\widehat{\F}^*_{2^d}$, it can be lifted from $\F_{2^d}$ to $\F_{2^m}$ by $\chi'=\chi\circ\N$ (see \cite[Theorem 5.28]{lidl}). Furthermore, $\chi$ is of order $2^d-1$ if and only if $\chi'$ is of order $2^d-1$. Then
\begin{align*}
  \sum_{j=0}^{2^d-2}
\overline{\chi'^j}(v)G(\chi'^j) &=\sum_{\chi\in\widehat{\F}^*_{2^d}}
\overline{\chi}(\N(v))G(\chi\circ\N)=
(-1)^{\frac{m}{d}-1}\sum_{\chi\in\widehat{\F}^*_{2^d}}
\overline{\chi}(v^{\frac{2^m-1}{2^d-1}})G(\chi)^{\frac{m}{d}}.
\end{align*}

Suppose $\delta=\beta^{\frac{2^m-1}{2^d-1}}\in\F_{2^d}^*$, then $\F_{2^d}^*=\bigcup\limits_{j=0}^{\frac{2^d-1}{t}-1}
\delta^j\langle\delta^{\frac{2^d-1}{t}}\rangle$. By Eq.~\eqref{eq:gu2}, we get
\begin{align*}
    |N|&=\frac{2^m-1}{2^d}+\frac{(-1)^{\frac{m}{d}-1}}{2^d}
\sum_{v\in\F^*_{2^d}}\sum_{\chi\in\widehat{\F}^*_{2^d}}
\overline{\chi}(v^{\frac{2^m-1}{2^d-1}})G(\chi)^{\frac{m}{d}}\\
&=\frac{2^m-1}{2^d}+\frac{(-1)^{\frac{m}{d}-1}}{2^d}
\sum_{\chi\in\widehat{\F}_{2^d}^*}G(\chi)^{\frac{m}{d}}
\sum_{j=0}^{\frac{2^d-1}{t}-1}\sum_{v\in\delta^j\langle\delta^{\frac{2^d-1}{t}}\rangle} \overline{\chi}(v^{\frac{2^m-1}{2^d-1}})\\
&=\frac{2^m-1}{2^d}+\frac{(-1)^{\frac{m}{d}-1}t}{2^d}
\sum_{\chi\in\widehat{\F}_{2^d}^*}G(\chi)^{\frac{m}{d}}
\sum_{j=0}^{\frac{2^d-1}{t}-1} \overline{\chi}(\delta^{j\frac{2^m-1}{2^d-1}}).
\end{align*}

Note that $\gcd(\frac{2^m-1}{t(2^d-1)},\frac{2^d-1}{t})=1$, then \begin{align*}\sum_{i=0}^{\frac{2^d-1}{t}-1} \overline{\chi}(\delta^{i\frac{2^m-1}{2^d-1}})
=\sum_{i=0}^{\frac{2^d-1}{t}-1} \overline{\chi}(\delta^{it})=\sum_{x\in\langle\delta^t\rangle} \overline{\chi}(x)=\begin{cases}
                     0, & \mbox{if } \chi\neq\chi_0; \\
                     \frac{2^d-1}{t}, & \mbox{if } \chi= \chi_0.
                   \end{cases}
\end{align*}

Hence we have
\begin{align*}
  |N| =\frac{2^m-2^d}{2^d}+\frac{(2^d-1)(-1)^{\frac{m}{d}-1}}{2^d}
\sum_{\chi\in(\widehat{\F}_{2^d}^*)^{\frac{2^d-1}{t}}\backslash\{\chi_0\}}
G(\chi)^{\frac{m}{d}}.
\end{align*}

 Note that $|G(\chi)|=2^{\frac{d}{2}}$ for $\chi\neq\chi_0$ (see \cite[Theorem~5.11]{lidl}), then we have
 \begin{align*}
	|N| \geq\frac{2^m-2^d-(2^d-1)(t-1)2^{\frac{m}{2}}}{2^d}.
\end{align*}

Note that $t\neq1$ and
\begin{align*}
	2^m-2^d-(2^d-1)(t-1)2^{\frac{m}{2}} = 2^m-t2^{\frac{m}{2}+d}+(t-1)2^{\frac{m}{2}}-2^d>2^m-t2^{\frac{m}{2}+d}.
\end{align*}

Since $t=\gcd(\frac{m}{d},2^d-1)\leq 2^d-1$, then $2^m-t2^{\frac{m}{2}+d}>2^m-2^{\frac{m}{2}+2d}\geq 0$ if $\frac{m}{d}\geq 4$.

 If $m=3d$, then $t=\gcd(3,2^d-1)=3$ and $d$ is even. One has
\[|N|=2^{3d}-2^d-(2^d-1)2^{\frac{3d}{2}+1}> 2^{3d}-2^{\frac{5d}{2}+1}\geq0. \]

If $m=2d$, then $t=\gcd(2,2^d-1)=1$, which contradicts to $t\neq1$.
Thus we complete the proof.
\end{proof}

Back to our situation, we have the following result.
\begin{lemma}\label{lemma:imp}
 Suppose $v_2(m)<v_2(i)$, then there exists  $a\in\F_{2^n}-\F_{2^m}$ such that
  $L_a(y)$ has  roots in $\F^*_{2^n}$.
\end{lemma}
\begin{proof}

For $v_2(m)<v_2(i)$, note that $d=\gcd(m,i)=\gcd(m+i, n)$. It suffices to show that there exists $a \in\F_{2^n}-\F_{2^m}$ such that  $L_a(y)$ has a root $y_0 \in \F_{2^m}^*$. Note that for $y\in\F_{2^m}^*$,
\begin{equation}\label{eq:bq1}
  L_a(y)=a^{2^i}y^{2^{2i}}+(a+a^{2^m})^{2^i}y^{2^{m+i}}+ay
=a^{2^i}y^{2^{2i}}+(a+a^{2^m})^{2^i}y^{2^{i}}+ay.
\end{equation}
Then we just need to find $(a, y)\in \F_{2^n}\times \F_{2^m}^*$   such that
\begin{equation}\label{eq:bq2}
  \begin{cases}
 a+a^{2^m}=y^{-2^i-1}v,\\
  (ay^{2^{i}+1})^{2^i}+ay^{2^{i}+1}=y^{2^{i}-2^{2i}}v,
\end{cases}
\end{equation}
for some $v\in \F_{2^d}^*$ (here $a\notin \F_{2^m}$ is automatic). Let $z=av^{-1}y^{2^{i}+1}$, then we just need to find $(z, y)\in \F_{2^n}\times \F_{2^m}^*$ such that
\begin{numcases}{}
z+z^{2^m}=1, \label{eqsystem1} \\
z^{2^i}+z=y^{2^{i}-2^{2i}}.\label{eqsystem2}
\end{numcases}

We consider Eq.~\eqref{eqsystem2}. Note that the $\F_{2^d}$-linear maps $\varphi_i:z\mapsto z^{2^i} +z$ and  $\varphi_d: z\mapsto z^{2^d} +z$ from $\F_{2^m}$ to itself have the same kernel $\F_{2^d}$ and $\varphi_i(z)=\varphi_d(z+z^{2^d}+\cdots+z^{2^{(\frac{i}{d}-1)d}})$, then
$\Im(\varphi_i)\subseteq \Im(\varphi_d)$ and hence $\Im(\varphi_d)=\Im(\varphi_i)$.  Note also that the group homomorphisms $y\mapsto y^{2^i(1-2^i)}$ and $y\mapsto y^{2^d-1}$ from $\F_{2^m}^*$ to itself have the same kernel and image. Then there is an one-to-one correspondents of  solutions $(z,y)\in \F_{2^m}\times \F_{2^m}^*$ of
 Eq.~\eqref{eqsystem2} and  of
  \begin{equation} \label{eqsystem3}
  	z^{2^d}+z = y^{2^d-1} .
  \end{equation}
 Eq.~\eqref{eqsystem3} is soluble if and only if there exists $y\in\F_{2^m}^*$ such that $\Tr_{2^m/2^d}(y^{2^d-1})=0$, which is guaranteed  by Lemma~\ref{lemma:gauss} as $d<m$ in this case.
Thus there exists $(z_0, y_0)\in \F_{2^m}^*\times \F_{2^m}^*$ such that $z^{2^i}_0+z_0=y^{2^{i}-2^{2i}}_0 $.

Let $w\in\F_{2^{2d}}\setminus\F_{2^d}$ satisfy $w^{2^d}+w=v_0$, then $w^{2^{m-d}}=w^{2^i}=w$ and $z=z_0+w\in\F_{2^n}\setminus\F_{2^m}$ is a solution of Eqs.~\eqref{eqsystem1} and ~\eqref{eqsystem2}. Thus, $y_0\in\F_{2^m}^*$ and $a=(z_0+w)y_0^{-2^{i}-1}v$ satisfy the equation $L_a(y)=0$.
\end{proof}

\begin{lemma}\label{lemma:pre}
For $0\leq i\leq m-1$, if $F(x)=x^{2^m+1}+x^{2^i+1}$ has $2^n-2^m$ bent components, then $v_2(m)\leq v_2(i)$.
\end{lemma}
\begin{proof}  Assume  $v_2(m)>v_2(i)$.  In this case  $d=\gcd(m,i)=\gcd(n,i)$, and $2d=\gcd(2i, m)$. This means $2^d-1=\gcd(2^m-1, 2^i-1)$ and $2^{2d}-1=\gcd(2^m-1, 2^{2i}-1)$, which then implies that $2^d+1$ is a factor of $2^m-1$ and thus prime to $2^m+1$.
	
Let $\alpha$ be a primitive element of $\F_{2^n}$. Let $a=\alpha^{k(2^m+1)}\in \F_{2^m}^*$ such that $a^{2^i-1}=\alpha^{(2^m+1)(2^d-1)}$.
By Proposition~\ref{proposition:p1}, for this $a$, $F_a(x)=\Tr_{2^n/2}(ax^{2^i+1})$
is not bent. By Lemma~\ref{lemma:diff},  $D_yF_a(x)=\Tr_{2^n/2}(x(ay^{2^i}+(ay)^{2^{n-i}}))$ is not balanced for some $y\in\F_{2^n}$, i.e., $a^{2^i-1} y^{2^{2i}-1}+1=0$ is soluble. Let
$a^{2^i-1}=\alpha^{(2^m+1)(2^d-1)}= y_0^{1-2^{2i}}$ and let $y_1\in \F_{2^n}^*$ such that
$ y_0^{1-2^{2i}}=y_1^{2^{2d}-1} $. Then the congruent equation
 \[  (2^{2d}-1) x\equiv  (2^d-1)(2^m+1) \bmod{(2^n-1)} \]
is soluble, equivalently, the equation
 \[  (2^d+1) x\equiv 2^m+1\bmod{(2^d+1)\cdot \frac{2^n-1}{2^{2d}-1}} \]
is soluble. This is not possible since  $2^d+1$ is prime to $2^m+1$.
\end{proof}

\begin{proof}[Proof of Theorem~\ref{theorem:bino}] If $i=0$, the result is trivial. We now assume $i\neq 0$ and $v_2(m)\leq v_2(i)$.

If $F(x)$ has maximal number of bent components, by Lemma~\ref{lemma:diff}, $F_a(x)$ is  bent function for all $a\in \F_{2^n}-\F_{2^m}$ and hence $D_yf_a(x)$ is balanced for any $y\in\F_{2^n}^*$. This implies $L_a(y)\neq0$ for all $y\in\F_{2^n}^*$.  Hence to show $F(x)$ does not have maximal number of bent components, it suffices to show there exists $a\in \F_{2^n}-\F_{2^m}$, such that $L_a(y)$ has a root in $\F_{2^n}^*$:
\begin{enumerate}[i]
	\item If $v_2(m)=v_2(i)$, this is implied by Lemma~\ref{lemma:bl}.
	
\item 	If $v_2(m)<v_2(i)$, this is implied
	by Lemma~\ref{lemma:imp}.
\end{enumerate}
Thus for $i\neq0$, $F(x)$ cannot have $2^n-2^m$ bent components.
\end{proof}

\begin{remark} For a general binomial vectorial function $F(x)=x^{d_1}+x^{d_2}$, our experimental result indicates that $F(x)$ is affine equivalent to
 $x^{2^m+1}$ or $x^{2^i}(x+x^{2^m})$ if $F(x)$ has maximal number of bent components, but so far we do not have a proof. We leave this as an open problem for future study.
\end{remark}

\section{Conclusion}We firstly study the  nonlinearity of $F(x)=x^{2^e}h(\Tr_{2^n/2^m}(x))$ with $2^n-2^m$ bent components, where $h(x)$ is a permutation over $\F_{2^m}$, and obtain the upper bound of the nonlinearity of $F$ based on the monomial $h$. Moreover, we give some plateaued and non-plateaued functions attaining the upper bound.
We secondly give two generic constructions of vectorial functions with maximal number of bent components, and obtain two new classes of such vectorial functions based on the Niho quadratic function and the Maiorana-MacFarland class. Moreover, our constructions  partially answer the open problem proposed by Pott et al. and contain vectorial functions outside the  complete Maiorana-MacFarland class.  We finally show that the binomial function $F(x)=x^{2^m+1}+x^{2^i+1}: \F_{2^{2m}}\rightarrow \F_{2^{2m}}$ has maximal number of bent components if and only if $i=0$.

\end{document}